%% file: main.tex
\documentclass{llncs}
\usepackage{url}
\usepackage{bm,amssymb,amsmath}

\makeatletter

\def\cramped  
  {\parskip\@outerparskip\@topsep\parskip  
  \@topsepadd2pt\itemsep3pt} 
\makeatletter

\providecommand{\newoperator}[3]{%
  \newcommand*{#1}{\mathop{#2}#3}}

\newoperator{\ProbOp}{\mathrm{Pr}}{\nolimits}

\newcommand\access{\ensuremath{\mathtt{access}}}
\newcommand\rank{\ensuremath{\mathtt{rank}}}

\newcommand\select{\ensuremath{\mathtt{select}}}

\newcommand\myparagraph[1]{\smallskip\noindent{\bf #1}}

\begin{document}


\setcounter{page}{1}
\title{Optimal Trade-Off for Succinct String Indexes}
 \author{
 R. Grossi\inst{1}
 \and A. Orlandi\inst{1}
 \and R. Raman\inst{2}
 }

\institute{Dipartimento di Informatica,
 Universit\`a di Pisa\\\email{{grossi,aorlandi}@di.unipi.it} \and
 Department of Computer Science, University of Leicester, United Kingdom\\\email{r.raman@mcs.le.ac.uk}}

\maketitle

\begin{abstract}
%
  Let $s$ be a string whose symbols are solely available through
  $\access(i)$, a read-only operation that \emph{probes}~$s$ and
  returns the symbol at position~$i$ in $s$. Many compressed data
  structures for strings, trees, and graphs, require two kinds of
  queries on $s$: $\select(c,j)$, returning the position in~$s$
  containing the $j$th occurrence of $c$, and $\rank(c,p)$, counting
  how many occurrences of $c$ are found in the first $p$ positions of
  $s$.  We give matching upper and lower bounds for this problem,
  improving the lower bounds given by Golynski [\emph{Theor.  Comput.
    Sci.}  387 (2007)] [PhD thesis] and the upper bounds of Barbay et
  al.\mbox{} [SODA 2007].  We also present new results in another
  model, improving on Barbay et al.\mbox{} [SODA 2007] and matching a
  lower bound of Golynski\mbox{} [SODA 2009].  The main contribution
  of this paper is to introduce a general technique for proving lower
  bounds on succinct data structures, that is based on the access
  patterns of the supported operations, abstracting from the
  particular operations at hand.
  For this, it may find application to other interesting problems on succinct data structures.
\end{abstract}


\section{Introduction}
\label{sec:intro}
\input{s1-intro}
\input{s2-lbmio}
\section{Upper bounds}
\label{sec:ub}
\input{s3-ub}
\small

\appendix
\input{apdx-a.tex}

\end{document}

%% file: s1-intro.tex

We are given a read-only sequence $s \equiv s\left[0,n-1\right]$ of
$n$ symbols over an integer alphabet $\Sigma=\left[\sigma\right]
\equiv \{0,1, \ldots, \sigma-1\}$, where $2 \le \sigma \le n$. 
The symbols in $s$ can be read using $\access(i)$, 
for $0 \leq i \leq n-1$: this primitive
\emph{probes}~$s$ and returns the symbol at position~$i$, denoted by
$s\left[i\right]$.
Given the sequence $s$, its
length~$n$, and the alphabet size~$\sigma$, we want to support the
following query operations for a symbol $c \in \Sigma$:
\begin{itemize}
\item $\select(c,j)$: return the position inside $s$ containing the $j$th
  occurrence of symbol $c$, or $-1$ if that occurrence does not exist;
\item $\rank(c,p)$: count how many occurrences of $c$ are found in
  $s\left[0,p-1\right]$.
\end{itemize}

We postulate that an auxiliary data structure, called a \emph{succinct index},
is constructed in a preprocessing step to help answer 
these queries rapidly. In this paper, we study the natural 
and fundamental \emph{time-space} 
tradeoff between two parameters $t$ and $r$ for this problem:
\begin{itemize}
\item $t =$ the \emph{probe complexity}, which is the maximal number of
  probes to $s$ (i.e.\mbox{} calls to $\access$) that the succinct
  index makes when answering a query\footnote{The time %
complexity of our results in the RAM model with %
logarithmic-sized words is linearly proportional to the probe %
complexity. Hence, we focus on the latter.};
\item $r =$ the \emph{redundancy}, which is the number of bits
  required by the succinct index, and does \emph{not} include the space
  needed to represent $s$ itself.
\end{itemize}
Clearly, these queries can be answered in negligible space
but $O(n)$ probes by scanning $s$, or in zero probes by
making a copy of $s$ in auxiliary memory at preprocessing time,
but with redundancy of $\Theta(n \log \sigma)$ bits.  We are interested in 
succinct indices that use few probes, and 
have redundancy $o(n \log \sigma)$, i.e., asymptotically
smaller than the space for $s$ itself. Specifically,
we obtain upper and lower bounds on the redundancy $r \equiv
r(t, n, \sigma)$, viewed as a function of the maximum number $t$ of probes,
the length $n$ of $s$, and the alphabet size $\sigma$.
We assume that $t > 0$ in the rest of the paper.

\subsubsection*{Motivation.} Succinct indices have
numerous applications to problems involving indexing massive 
data sets \cite{BHM07indexes}.  The $\rank$ and $\select$
operations are basic primitives at the heart of 
many sophisticated indexing data structures for strings, 
trees,  graphs, and sets \cite{GMR06large}. 
Their efficiency is crucial to make these indexes fast and space-economical.
Our results are most interesting for the case of
``large'' alphabets, where $\sigma$ is a not-too-slowly
growing function of $n$.  Large alphabets are common in
modern applications: e.g. many files are in
Unicode character sets, where $\sigma$ is of the
order of hundreds or thousands. Inverted lists or documents in
information retrieval systems can be seen as sequences $s$ of words,
where the alphabet $\Sigma$ is obviously large and increasing with
the size of the collection (it is the vocabulary 
of distinct words appearing over the entire document repository).

\myparagraph{Our results.}
Our first contribution is showing that the redundancy $r$ in bits 
\begin{equation}
  \label{eq:redundancy}
  r(t,n,\sigma) = \Theta\left(\frac{ n\log \sigma } {t}\right)
\end{equation}
is tight for any succinct index solving our problem, for $t = O(\log
\sigma / \log\log \sigma)$. (All the logarithms in this paper are to
the base~2.) We provide matching upper and lower bounds for this range
of values on $t$, under the assumption that $O(t)$ probes are allowed
for $\rank$ and $\select$, i.e.\mbox{} we ignore multiplicative
constant factors. 
The result is composed by a lower bound of $r = \Omega(\frac{n \log\sigma} t)$
bits that holds for $t = o(\log \sigma)$ and by an upper bound of
$r = O(\frac{n \log \sigma} t + n \log \log\sigma)$. 
We also provide a lower bound of $r = \Omega(\frac{n \log t}{t})$ for $t = O(n)$, 
thus leaving open what the optimal redundancy when $t = \Omega( \frac{\log \sigma}{\log\log\sigma} )$.
Running times for the upper bound are $O(t + \log \log \sigma)$ for $\rank$ and $O(t)$ for $\select$.

An interpretation of~\eqref{eq:redundancy} is that, given a data
collection $D$, if we want to build an additional succinct index on $D$
that saves space by a factor $t$ over that taken by $D$, we have to
pay $\Omega(t)$ access cost for the supported queries. 
Note that the plain storage of the
sequence $s$ itself requires $n \log \sigma$ bits. 
Moreover, our
result shows that it is suboptimal to build $\sigma$ individual
succinct indexes (like those for the binary-alphabet case,
e.g.~\cite{RRR07}), one per symbol $c \in \left[\sigma\right]$: the
latter approach has redundancy $\Theta(\frac{\sigma n \log t}{t})$
while the optimal redundancy is given in eq.~(\ref{eq:redundancy}),
when $t = O(\log \sigma / \log\log \sigma)$.

Lower bounds are our main findings, while the matching upper bounds
are derived from known algorithmic techniques.  Thus, our second
contribution is a general technique that extends the algorithmic
encoding/decoding approach in~\cite{DLO03linear} in the sense 
that it abstracts from the specific query operation at hand, and focuses on its access
pattern solely. For this, we can single out a sufficiently large,
conflict free subset of the queries that are classified as
\emph{stumbling} or $z$-\emph{unique}. In the former case, we extract
direct knowledge from the \emph{probed} locations; in the latter, the
novelty of our approach is that we can extract (implicit) knowledge
also from the \emph{unprobed} locations. We are careful not to exploit
the specific semantics of the query operations at this stage. As a
result, our technique applies to other kinds of query operations for
predecessor, prefix sum, permutation, and pattern searching problems,
to name a few, as long as we can extract a sufficiently large subset
of the queries with the aforementioned features. We will
discuss them extensively in the full version.


We also provide further running times for the rank/select 
problem. For example, if
$\sigma = (\log n)^{O(1)}$, the $\rank$ operation requires only $O(t)$
time; also, we can get $O(t \log \log \sigma \log^{(3)} \sigma)$ time\footnote{We define %
  $\log^{(1)} x := \log_2 x$ and for integer $i \ge 2$, %
  $\log^{(i)} x := \log_2 (\log^{(i-1)} x)$.}  for $\rank$ and $O(t
\log \log \sigma)$ time for $\select$ (Theorem~\ref{thm:ub}). We also
have a lower bound of $r = \Omega\bigl(\frac{n \log t}{t}\bigr)$ bits
for the redundancy when $1 \leq t \leq n/2$, which leaves open what is the optimal
redundancy when $t = \Omega(\log \sigma)$. 
As a corollary, we can obtain an entropy-compressed
data structure that represents~$s$ using $nH_k(s) + O(\frac{n\log
  \sigma}{\log \log \sigma})$ bits, for any $k = o(\frac{\log_\sigma
  n}{\log \log \sigma})$, supporting $\access$ in $O(1)$ time, $\rank$
and $\select$ in $O(\log \log \sigma)$ time (here, $H_k(s)$ is the
$k$th-order empirical entropy).

%


\myparagraph{Related work.}
Succinct data structures are generally divided into two kinds,
\emph{systematic} and \emph{non-systematic} \cite{GM07}.
Non-systematic data structures encode the input data and any auxiliary
information together in a single representation, while systematic
don't. The concept of succinct indexes applies to systematic ones;
moreover, the concept of probes does not apply to non-systematic ones,
since the input data is not distinguished from the index.
In terms of time-space trade-off, our results extend the complexity
gap between systematic and non-systematic succinct data structures
(which was known for $\sigma=2$) to any integer alphabet of size
$\sigma \leq n$.  This is easily seen by considering the case of
$O(1)$ time/probes for $\select$. Our systematic data structure
requires $r=O(n \log \sigma)$ bits of redundancy whereas the
non-systematic data structure of~\cite{GMR06large} uses just $O(n)$
bits of redundancy.  However, if the latter should also provide
$O(1)$-time $\access$ to the encoded string, then its redundancy
becomes $O(n \log \sigma)$.  Note that eq.~(\ref{eq:redundancy}) is
targeted for non-constant alphabet size $\sigma$ whereas, for constant
size, the lower and upper bounds for the $\sigma = 2$ case
of~\cite{golynski07} can be extended to obtain a matching bound of
$\Omega(\frac{n \log t}{t})$ bits (see Appendix~\ref{apdx:extending}).


The conceptual separation of the index from the input data was
introduced to prove lower bounds in \cite{GM07}. It was then
explicitly employed for upper bounds in
\cite{FerraginaV07,GonzalezN06,SadakaneG06}, and was fully formalized
in \cite{BHM07indexes}. The latter contains the best known upper
bounds for our problem\footnote{We compare ourselves with the improved %
bounds given in the full version of \cite{BHM07indexes}.},
i.e.\mbox{} $O(s)$ probes for $\select$ and $O(s
\log k)$ probes for $\rank$, for any two parameters $s \leq \log
\sigma/\log\log \sigma$ and $k \leq \sigma$, with redundancy $O(n\log
k + n (1/s+1/k) \log \sigma)$. For example, fixing $s=k=\log\log
\sigma$, they obtain $O(\log\log \sigma)$ probes for $\select$ and
$O(\log \log \sigma \log^{(3)} \sigma)$ probes for $\rank$, with
redundancy $O(n \log \sigma / \log \log \sigma)$. By
eq.~(\ref{eq:redundancy}), we get the same redundancy with $t=O(\log
\log \sigma)$ probes for both $\rank$ and $\select$. Hence, our probe
complexity for $\rank$ is usually better than \cite{BHM07indexes}
while that of $\select$ is the same.  Our $O(\log \log \sigma)$
running times are all better when compared to $O((\log \log \sigma)^2
\log^{(3)} \sigma)$ for $\rank$ and $O((\log \log \sigma)^2)$ for
$\select$ in \cite{BHM07indexes}.

%% file: s2-lbmio.tex
\section{General Technique}
\label{sec:technique}
This section aims at stating a general lower bound technique, of
independent interest, which applies not only to both $\rank$ and
$\select$  but to other query operations as well.  Suppose we
have a set $S$ of strings of length $n$, and a set $Q$ of queries that
the must be supported on $S$ using at most $t$ probes each and an
unknown amount $r$ of redundancy bits.  Under certain assumptions on
$S$ and $Q$, we can show a lower bound on $r$.  Clearly, any choice of
$S$ and $Q$ is allowed for the upper bound.

\myparagraph{Terminology. }We now give a framework that relies on a
simple notion of entropy $H(S)$, where $H(X) = \lceil \log |X| \rceil$
for any class of $|X|$ combinatorial objects~\cite{CoverThomas}.  The
framework extends the algorithmic encoding/decoding approach
\cite{DLO03linear}.  Consider an arbitrary algorithm $\mathcal A$ that
can answer to any query in $Q$ performing at most $t$ probes on any $s
\in S$, using a succinct index with $r$ bits.  We describe how to
encode $s$ using $\mathcal A$ and the succinct index as a black box,
thus obtaining $E(s)$ bits of encoding. Then, we describe a decoder
that knowing $\mathcal A$, the index of $r$ bits, and the latter
encoding of $E(s)$ bits, is able to reconstruct $s$ in its original
form.  The encoding and decoding procedure are allowed unlimited (but
finite) computing time, recalling that $\mathcal A$ can make at most
$t$ probes per query.

The lower bound on $r$ arises from the necessary condition $\max_{s\in
  S} E(s) + r \geq H(S)$, since otherwise the decoder cannot be
correct. Namely, $r \geq H(S) - \max_s E(s)$: the lower $E(s)$, the
tighter the lower bound for $r$.  Our contribution is to give
conditions on $S$ and $Q$ so that the above approach can hold for a
variety of query operations, and is mostly oblivious of the specific
operation at hand since the query access pattern to $s$ is
relevant. This appears to be novel.

First, we require $S$ to be sufficiently dense, that is, $H(S) \geq n \log \sigma - \Theta(n)$.
Second, $Q$ must be a subset of $\left[\sigma\right] \times \left[n\right]$, so that
the first parameter specifies a character $c$ and the second one an integer $p$.
Elements of $Q$ are written as $q_{c,p}$. Third, answers to queries must be within $[n]$.
The set $Q$ must contain a number of stumbling or $z$-unique queries, as we define now.
Consider an execution of $\mathcal A$ on a query $q_{c,p} \in Q$ for a string $s$. The set of
accessed position in $s$, expressed as a subset of $\left[n\right]$ is called
an \emph{access pattern}, and is denoted by $\text{Pat}_s(q_{c,p})$. 


First, \emph{stumbling} queries imply the occurrence of a certain symbol $c$ inside their
own access pattern: the position of $c$ can be decoded by using just the answer and
the parameters of the query.
Formally, $q_{c,p} \in Q$ is stumbling if there exists a computable function $f$
that takes in input $c, p$ and the answer of $q_{c,p}$ over $s$, and outputs
a position $x \in \text{Pat}_s(q_{c,p})$ such that $s[x] = c$.
The position $x$ is called the \emph{target} of $q_{c,p}$.
The rationale is that the encoder does not need to store any information
regarding $s[x] = c$, since $x$ can be extracted by the decoder from
$f$ and the at most $t$ probed positions by $\mathcal A$.
We denote by $Q'_s \subseteq Q$ the set of stumbling queries over $s$.

Second, \emph{$z$-unique} queries are at the heart of our technique, where
$z$ is a positive integer. Informally, they have specific answers
implying unique occurrences of a certain symbol $c$ in a segment of
$s$ of length $z+1$.  Formally, a set $U$ of answers is
\emph{$z$-unique} if for every query $q_{c,p}$ having answer in $U$,
there exists a \emph{unique} $i \in [p,p+z]$ such that $s[i] = c$
(i.e.\mbox{} $s[j] \neq c$ for all $j \in [p,p+z], j \neq i$). 
A query $q_{c,p}$ having answer in $U$ is called $z$-unique and the
corresponding position $i$ is called the \emph{target} of $q_{c,p}$.
Note that, to our purposes, we will restrict to the cases where $H(U) = O(n)$.
The rationale is the
following: when the decoder wants to rebuild the string it must
generate queries, execute them, and test whether they are $z$-unique by
checking if their answers are in $U$.  Once that happens, it can infer
a position $i$ such that $s[i] = c$, even though such a position is
not probed by the query.  We denote by $Q''_s(z) \subseteq Q \setminus
Q'_s$ the set of $z$-unique queries over $s$ that are \emph{not}
stumbling.
We also let $\text{Tgt}_s(q_{c,p})$ denote the target of query $q_{c,p}$ over $s$, if it exists, and
let $\text{Tgt}_s(Q) = \cup_{q \in Q} Tgt_s(q)$ for any set of queries $Q$.

\myparagraph{Main statement. }We now state our main theorem.
Let $S$ be a set of strings such that $H(S) \geq n\log\sigma - \Theta(n)$.
Consider a set of queries $Q$ that can be answered by performing at most $t$ probes per query and using $r$ bits
of redundancy.
\begin{theorem}
\label{thm:generalbound}
For any $z \in \left[\sigma\right]$, let $\lambda(z) = \min_{s \in S} |\text{Tgt}_s(Q'_s) \cup \text{Tgt}_s(Q''_s(z))|$.
Then, there exists integers $\gamma$ and $\delta$ with $\min\{\lambda(z), n\}/(15t) \leq \gamma + \delta \leq \lambda(z)$,
such that any succinct index has redundancy 
$$r \geq \gamma\log\left(\frac \sigma z \right) + \delta \log\left(\frac{\sigma\delta}{t|Q|}\right) - \Theta(n)$$
\end{theorem}
The proof goes through a number of steps, each dealing with a different issue and is deferred to Section~\ref{sec:proving}.

\myparagraph{Applications. } 
We now apply Theorem~\ref{thm:generalbound} to our two main problems, for an alphabet size $\sigma \leq n$.
\begin{theorem}
  \label{thm:rankbound}
  Any algorithm solving
  $\rank$ queries on a string $s \in \left[\sigma\right]^n$ using at most $t = o (\log \sigma)$ character probes (i.e.\mbox{}
  $\access$ queries), requires a succinct
  index with $r = \Omega\bigl(\frac { n \log \sigma} {t} \bigr)$ bits
  of redundancy.
\end{theorem}
\begin{proof}
We start by defining the set $S$ of strings.
For the sake of presentation, suppose $\sigma$
divides $n$. An arbitrary string $s \in S$ is the concatenation of
$n/\sigma$ \emph{permutations} of $\left[\sigma\right]$. Note that $|S| =
(\sigma!)^{n/\sigma}$ and so we have
$H(S) \geq n\log \sigma - \Theta(n)$ bits
(by Stirling's approximation).

Without loss of generality, we prove the bound on a derivation of the
$\rank$ problem. We define the set $Q$ and fix the parameter $z = 
\sigma^{3/4}\sqrt{t}$, so that the queries are $q_{c,p} = \rank(c,p+z)
- \rank(c,p)$, where $c \in \left[\sigma\right]$ and $p \in
\left[n\right]$ with $p \bmod z \equiv 0$.  In this setting, the
$z$-unique answers are in $U = \{1\}$.  Indeed, whenever $q_{c,p} =
1$, there exists just one instance of $c$ in $s[p,p+z]$.  Note
that $|Q| = n \sigma / z > n$, for $\sigma$ larger than some constant.

Observe that $\lambda(z) \geq n$, as each position $i$ in $s$ such
that $s[i] = c$, is the target of exactly one query $q_{c,p}$:  supposing
the query is not stumbling, such a query is surely $z$-unique.  By
Theorem~\ref{thm:generalbound}, $\gamma + \delta \geq n/(30t)$ since
 a single query is allowed to make up to $2t$ probes now. (This
causes just a constant multiplicative factor in the lower bound.)

Having met all requirements, we apply Theorem~\ref{thm:generalbound},
and get
\begin{equation}
  \label{eq:rank}
  r \geq \gamma \log\left(\frac \sigma z\right) - \delta \log\bigl(\frac{n t}{z \delta}\bigr)
\end{equation}

We distinguish between two cases.  If $\delta \leq n / \sigma^{1/4}$,
then $\delta \log ( (nt)/ (z\delta)  ) \leq \frac n {\sigma^{1/4}}
\log \bigl( \frac{nt\sigma^{1/4}}{n z} \bigr) \leq \frac{n}{2
  \sigma^{1/4}} \log (t/\sigma)$, since $\delta \log ( 1 / \delta )$
is monotone increasing in $\delta$ as long as $\delta \leq
\lambda(z)/2$ (and $n / \sigma^{1/4} \leq \lambda(z)/2$ for
sufficiently large $\sigma$).  Hence, recalling that $t = o(\log
\sigma)$, the absolute value of the second term on the right hand
of~\eqref{eq:rank} is $o(n/t)$ for $\sigma$ larger than a constant.
Moreover, $\gamma \geq n/(30t) - \delta \geq n / (60t)$ in this
setting, so that the bound in~\eqref{eq:rank} reduces to
$$r \geq \frac n {240t} \log \sigma - \frac n {120t} \log t - \Theta(n) = \frac n {240t} \log \sigma - \Theta(n).$$

In the other case, we have $\delta \geq n / \sigma^{1/4}$, and $\delta\log ( (nt)/(z\delta) ) \leq \delta \log \bigl( \frac{\sigma^{1/4} t } {\sigma^{3/4} \sqrt{t}} \bigr)= \frac {\delta}{2} \log(t/\sigma)$.
Therefore, we know in~\eqref{eq:rank} that $\gamma\log(\sigma/z) +
(\delta/2) \log (\sigma / t) \geq \frac 1 2 (\gamma+\delta)
\log(\sigma / z)$, as we chose $z \geq t$.
Again, we obtain
$$r \geq \frac n {120t} \log \sigma - \Theta(n).$$
In both cases, the $\Theta(n)$ term is negligible as $t = o(\log
\sigma)$, hence the bound.\qed
\end{proof}


\begin{theorem}
  \label{thm:selectbound}
  Any algorithm solving
  $\select$ queries on a string $s \in \left[\sigma\right]^n$ using at most $t = o (\log \sigma)$ character probes (i.e.\mbox{}
  $\access$ queries), requires a succinct
  index with $r = \Omega\bigl(\frac { n \log \sigma} {t} \bigr)$ bits
  of redundancy.
\end{theorem}
\begin{proof}
The set $S$ of strings
is composed by \emph{full strings}, assuming that $\sigma$ divides~$n$.
A full string contains each character exactly $n / \sigma$
times and, differently from Theorem~\ref{thm:rankbound}, has no restrictions on where they
can be found.
Again, we have $H(S) \geq n \log \sigma - \Theta(n)$.

The set $Q$ of queries is $q_{c,p} = \select(c,p)$, where $p \in [n/\sigma]$, and all queries in
$Q$ are stumbling ones, as $\select(c,i) = x$ immediately implies that $s[x] = c$
(so $f$ is the identity function).
There are no $z$-unique queries here, so we can fix any value of $z$:
we choose $z=1$. 
It is immediate to see that $\lambda(z) = n$, and $|Q| = n$, as there are only
$n/\sigma$  queries for each symbols in $[\sigma]$.
By Theorem~\ref{thm:generalbound}, we know that $\gamma+\delta \geq n/(15t)$.
Hence, the bound is
$$r \geq \gamma \log\sigma + \delta \log \bigl(\frac{\sigma \delta}{n t}\bigr) \geq \frac n {15t} \log (\sigma / t^{2}) - \Theta(n).$$

\noindent Again, as $t = o(\log\sigma)$ the latter term is negligible and the bound follows.
\qed
\end{proof}



\section{Proof of Theorem \ref{thm:generalbound}}
\label{sec:proving}
We give an upper bound on $E(s)$ for any $s \in S$ by describing an 
encoder and a decoder for $s$. In this way we can use the relation $\max_{s \in S} E(s) + r \geq H(S)$
to induce the claimed lower bound on $r$ (see Section~\ref{sec:technique}).
We start by discussing how we can use
$z$-unique and stumbling queries to encode a single position and its content compactly. 
Next, we will
deal with conflicts between queries: not all queries in $Q$ are useful for encoding.
We describe a mechanical way to select a sufficiently large subset of $Q$ so that
conflicts are avoided. Bounds on $\gamma$ and $\lambda$ arise from such a process.
To complete the encoding, we present how to store the parameters of the queries
that the decoder must run. 

\myparagraph{Entropy of a single position and its content. }
We first evaluate the entropy of positions and their contents by exploiting the knowledge of $z$-unique and stumbling queries.
We use the notation $H(S|\Omega)$ for some event $\Omega$ as a shortcut for $H(S')$ where $S' = \{ s \in S | s \text{ satisfies } \Omega\}$.

\begin{lemma}
\label{lem:zEntropy}
For any $z \in \left[\sigma\right]$, let  $\Omega_{c,p}$
be the condition ``$q_{c,p}$ is $z$-unique''. Then it holds $H(S)-H(S|\Omega_{c,p}) \geq \log(\sigma/z) - O(1)$.
\end{lemma}
\begin{proof}
Note that set $(S|\Omega_{c,p}) = \{ s \in S: \text{Tgt}_s(q_{c,p}) \text{ is
  defined on } s\}$ for a given query $q_{c,p}$.
It is $|(S|\Omega_{c,p})| \leq (z+1) \sigma^{n-1}$ since there at most $z+1$ candidate target cells compatible with $\Omega_{c,p}$
and at most $|S|/\sigma$ possible strings with position containing~$c$ at a fixed position.
So, $H(S|\Omega_{c,p}) \leq \log(z+1) + H(S) - \log\sigma$, hence the bound.
\qed
\end{proof}

\begin{lemma}
\label{lem:stumblingEntropy}
Let $\Omega'_{c,p}$ be the condition
``$q_{c,p}$ is a stumbling query''.  Then, it holds that $H(S)-H(S|\Omega'_{c,p}) \geq \log(\sigma/t) - O(1)$.
\end{lemma}
\begin{proof}
The proof for this situation is already known from~\cite{golynski-thesis}.
In our notation, the proof goes along the same lines as that of Lemma~\ref{lem:zEntropy},
except that we have $t$ choices instead of $z+1$.
To see that, let $m_1, m_2, \ldots, m_t$ be the positions, in temporal order, probed by the algorithm $\mathcal A$ 
on $s$ while answering $q_{c,p}$.
Since the query is stumbling, the target will be one of $m_1, \ldots, m_t$. It suffices
to remember which one of the $t$ steps probe that target, since their values $m_1, \ldots, m_t$ are
deterministically characterized given $\mathcal A$, $s$, $q_{c,p}$.  \qed
\end{proof}

\myparagraph{Conflict handling. } 
In general, multiple instances of Lemma~\ref{lem:zEntropy} and/or Lemma~\ref{lem:stumblingEntropy}
cannot be applied independently.
We introduce the notion of conflict on the targets and show how to circumvent this difficulty.
Two queries $q_{b,o}$ and $q_{c,p}$ \emph{conflict} on $s$ if at least one
of the following three condition holds:
$(i)$ $\text{Tgt}_s(q_{c,p}) \in \text{Pat}_s(q_{b,o})$, $(ii)$ $\text{Tgt}_s(q_{b,o}) 
\in \text{Pat}_s(q_{c,p})$, $(iii)$ $\text{Tgt}_s(q_{c,p}) = \text{Tgt}_s(q_{b,o})$.
A set of queries where no one conflicts with another is called \emph{conflict free}.
The next lemma is similar to the one found in~\cite{golynski09}, but the context is different.

Lemma~\ref{lem:iterprocess} defines a lower bound on the maximum size of a
conflict free subset of $Q$. We use an iterative procedure that maintains at each $i$th step
a set $Q^*_i$ of conflict free queries and a set $C_i$ of available targets, such that
no query $q$ whose target is in $C_i$ will conflict with any query $q' \in Q^*_{i-1}$.
Initially, $C_0$ contains all targets for the string $s$, so that by definition $|C_0| \geq \lambda(z)$.
Also, $Q^*_0$ is the empty set.
\begin{lemma}
\label{lem:iterprocess}
Let $i \geq 1$ be an arbitrary step and assume $|C_{i-1}| > 2|C_0|/3$. Then, there
exists $Q^*_i$ and $C_i$ such that (a) $|Q^*_i| = 1 + |Q^*_{i-1}|$, (b) $Q^*_i$ is
conflict free, (c) $|C_i| \geq |C_0| - 5it \geq  \lambda(z) - 5it$.
\end{lemma}
\begin{proof}
We first prove that there exists $u \in C_{i-1}$ such that no more than $3t$ queries probe $u$. Assume by contradiction
that for any $u$, at least $3t$ queries probe $u$. Then, we would collect $3t|C_{i-1}| > 2|C_0|t$ probes in total. However,
any query can probe at most $t$ cells, summing up to $|C_0|t$, giving a contradiction.
At step $i$, we choose $u$ as a target, say, of query $q_{c,p}$ for
some $c,p$. This maintains invariant (a) as $Q^*_i = Q^*_{i-1} \cup
\{q_{c,p}\}$. As for invariant (b), we remove the potentially conflicting targets from $C_{i-1}$, and produce $C_{i}$. Let $I_{u} \subseteq C_{i-1}$ be
the set of targets for queries probing $u$ over $s$, where by the above properties $|I_u| \leq 3t$.
We remove $u$ and the elements in $I_u$ and $\text{Pat}_s(q_{c,p})$. 
So, $|C_{i}| = |C_{i-1}| - |\{u\}| - |I_{u}| - |\text{Pat}_s(q_{c,p})| \geq |C_{i-1}| - 1 - 3t - t \geq |C_0| - 5it.$
\qed
\end{proof}
By applying Lemma~\ref{lem:iterprocess} until $|C_i| \leq 2|C_0|/3 $, we obtain a final
set $Q^*$, hence the following:
\begin{corollary}
\label{cor:qset}
For any $s \in S$, $z \in [\sigma]$, there exists a set $Q^*$ containing $z$-unique and stumbling queries of size $\gamma+\delta \geq \min\{\lambda(z),n\}/(15t)$, where $\gamma = |\{ q \in Q^* | q \text{ is stumbling on } s\}|$ and $\delta = |\{ q \in Q^* | q \text{ is $z$-unique on } s\}|$.
\end{corollary}

\myparagraph{Encoding. }
We are left with the main task of describing the encoder. Ideally, we would like to
encode the targets, each with a cost as stated in Lemma~\ref{lem:zEntropy} and Lemma~\ref{lem:stumblingEntropy},
for the conflict free set $Q^*$ mentioned in Corollary~\ref{cor:qset}. Characters in the remaining positions can be encoded
naively as a string. This approach has a drawback.
While encoding which queries in $Q$ are stumbling has a payoff when compared to Lemma~\ref{lem:stumblingEntropy},
we don't have such a guarantee for $z$-unique queries when compared to Lemma~\ref{lem:zEntropy}.
Without getting into details, according to the choice of the parameters $|Q|$, $z$ and $t$, such encoding sometimes
saves space and sometimes does not: it may use even more space than $H(S)$.
For example, when $|Q| = O(n)$, 
even the naive approach works and yields
an effective lower bound. Instead, if $Q$ is much larger, savings are not guaranteed.
The main point here is that we want to overcome such a dependence on the parameters and always guarantee
a saving, which we obtain by means of an implicit encoding of $z$-unique queries.
Some machinery is necessary to achieve this goal. 

\myparagraph{Archetype and trace. } Instead of trying to directly encode the information of $Q^*$ as discussed above, 
we find a query set $Q^A$ called the \emph{archetype} of $Q^*$, that is indistinguishable
from $Q^*$ in terms of $\gamma$ and $\delta$. The extra property of $Q^A$ is to be
decodable using just $O(n)$ additional bits, hence $E(s)$ is smaller when $Q^A$ is employed.
The other side of the coin is that our solution requires a two-step encoding.
We need to introduce the concept of \emph{trace} of a query $q_{c,p}$
over $s$, denoted by $\text{Trace}_s(q_{c,p})$. Given the access pattern
$\text{Pat}_s(q_{c,p}) = \{ m_1 <  m_2 < \cdots < m_t \}$ (see Section~\ref{sec:technique}), the trace is
defined as the string $\text{Trace}_s(q_{c,p}) = s[m_1] \cdot s[m2] \cdots \cdot s[m_t]$.
We also extend the concept to sets of queries, so that for $\widehat Q \subseteq Q$, we have $\text{Pat}_s(\widehat Q) = \bigcup_{q \in \widehat Q} \text{Pat}_s(q)$, and $\text{Trace}_s(\widehat Q)$ is defined using the sorted positions in $\text{Pat}_s(\widehat Q)$.

Then, we define a \emph{canonical ordering} between query sets. We define the predicate $q_{c,p} \prec q_{d,g}$
iff $p < g$ or $p = g \land c < d$ over queries, so that we can sort queries inside a single query set.
Let $Q_1 = \{ q_1 \prec q_2 \prec \cdots \prec q_x \}$ and let $Q_2 = \{ q'_1 \prec q'_2 \prec \cdots \prec q'_y \}$
be two distinct  query sets. We say that $Q_1 \prec Q_2$ iff either
$q_1 \prec q'_1$ or recursively $(Q_1 \setminus \{q_1\} ) \prec (Q_2 \setminus \{q'_1\} )$.

Given $Q^*$, its archetype $Q^A$ obeys to the following conditions for the given~$s$:
\begin{itemize}
\item it is conflict free and has the same number of queries of $Q^*$;
\item it contains exactly the same stumbling queries of $Q^*$, and all remaining queries are $z$-unique 
(note that they may differ from those in $Q^*$);
\item if $p_1, p_2, \ldots, p_x$ are the positional arguments of queries in $Q^*$, then the same
positions are found in $Q^A$ (while character $c_1, c_2, \ldots, c_x$ may change);
\item $\text{Pat}_s(Q^*) = \text{Pat}_s(Q^A)$;
\item among those query sets complying with the above properties, it is the minimal w.r.t.\mbox{} to the canonical ordering $\prec$.
\end{itemize}
Note that $Q^*$ complies with all the conditions above but the last. Therefore, the archetype of $Q^*$ always exists, being either
a smaller query set (w.r.t.\mbox{} to $\prec$)  or $Q^*$ itself.
The encoder can compute $Q^A$ by exhaustive search, since its time complexity is not relevant to the lower bound.

\myparagraph{First step: encoding for trace and stumbling queries. }
As noted above the stumbling queries for $Q^*$ and $Q^A$ are the same,
and there are $\delta$ of them.
Here, we encode the trace together with the set of stumbling queries. The rationale
is that the decoder must be able to rebuild the original trace only, whilst encoding of the
positions which are not probed is left to the next step, together with $z$-unique queries.
Here is the list of objects to be encoded in order:
\begin{enumerate}
\item[$(a)$] The set of stumbling queries expressed as a subset of $Q$.
\item[$(b)$] The access pattern $\text{Pat}_s(Q^A)$ encoded as a subset of $[n]$, the positions
of~$s$.
\item[$(c)$] The \emph{reduced} trace, obtained from $\text{Trace}_s(Q^A)$ by 
removing all the characters in positions that are targets of stumbling queries. 
Encoding is performed naively by storing each character
using $\log \sigma$ bits. The positions thus removed, relatively to the trace, are stored as a subset of $[|\text{Trace}_s(Q^A)|]$.
\item[$(d)$] For each stumbling query $q_{c,p}$, in the canonical order, an encoded integer $i$ of $\log t$ bits indicating
that the $i$th probe accesses the target of the query. 
\end{enumerate}

The decoder starts with an empty string, it reads the access pattern in $(b)$, the set of removed positions
in $(c)$, and distributes the contents of the reduced trace $(c)$ into the remaining positions.
In order the fill the gaps in $(c)$, it recovers the stumbling queries in $(a)$ and runs each of them, in canonical
ordering. Using the information in $(d)$, as proved by Lemma~\ref{lem:stumblingEntropy}, it can discover the target
in which to place its symbol $c$.
Since $Q^A$ is conflict free, we are guaranteed that each query will always find a symbol in the probed positions.

\begin{lemma}
\label{lem:traceencode}
Let $\ell$ be the length of $\text{Trace}_s(Q^A)$. The first step encodes information $(a)$--$(d)$ using  
at most
$\ell \log \sigma + O(n) + \delta\log(|Q|/\delta) - \delta\log(\sigma/t)$ bits.
\end{lemma}
\begin{proof}
Space occupancy for all objects:~$(a)$ uses $\log{|Q|\choose\delta} = \delta\log(|Q|/\delta) + O(\delta)$;
~$(b)$ uses $\log{n\choose\ell} \leq n$ bits;~$(c)$ uses $(\ell-\delta)\log\sigma$ bits for the reduced trace plus at most $\ell$
bits for the removed positions;~$(d)$ uses $\delta\log t$ bits.\qed
\end{proof}

\myparagraph{Second step: encoding of $z$-unique queries and unprobed positions. }
We now proceed to the second step, where targets for $z$-unique queries are encoded along
with the unprobed positions. They can be rebuilt using queries in $Q^A$. To this end, we assume
that encoding of Lemma~\ref{lem:traceencode} has already been performed and,
during decoding, we assume that the trace has been already rebuilt.
Recall that $\gamma$ is the number of $z$-unique queries.
Here is the list of objects to be encoded:
\begin{enumerate}
\item[$(e)$] The set of queries in $Q^A$ that are $z$-unique, expressed as a subset of $Q^A$ according
to the canonical ordering $\prec$. Also the set of $z$-unique answers $U$ is encoded as a subset of $[n]$.
\item[$(f)$] For each $z$-unique query $q_{c,p}$, in canonical order,
  the encoded integer $p$. This gives a multiset of  $\gamma$ integers in $[n]$.
\item[$(g)$] The \emph{reduced} unprobed region of the string, obtained
by removing all the characters in positions that are targets of $z$-unique queries. 
Encoding is performed naively by storing each character
using $\log \sigma$ bits. The positions thus removed, relatively to the unprobed region, are stored as a subset of $[n-\ell]$.
\item[$(h)$] For each $z$-unique query $q_{c,p}$, in the canonical order, an encoded integer $i$ of $\log z + O(1)$ bits indicating
which position in $[p,p+z]$ contains $c$.
\end{enumerate}

The decoder first obtains $Q^A$ by exhaustive search.
It initializes a set of $|Q^A|$ empty couples $(c,p)$ representing the arguments of each query in canonical order.
It reads $(e)$ and reuses $(a)$ to obtain the parameters of the stumbling queries inside $Q^A$.
It then reads $(f)$ and fills all the positional arguments of the queries.
Then, it starts enumerating all query sets in canonical order that are compatible with the arguments known so far.
That is, it generates characters for the arguments of $z$-unique queries, since the rest is known.
Each query set is then tested in the following way. The decoder executes each query by means of the trace.
If the execution tries a probe outside the access pattern, the decoder skips to the next query set.
If the query conflicts with any other query inside the same query set, the decoder skips.
If the query answer denotes that the query is not $z$-unique (see Section~\ref{sec:technique} and $(e)$), it skips.
In this way, all the requirements for the archetype are met, hence the first
query set that is not skipped is $Q^A$.

Using $Q^A$ the decoder rebuilds the characters in the missing positions of the reduced unprobed region: it starts by reading
positions in $(g)$ and using them to distribute the characters in the reduced region encoded by $(g)$ again.
For each $z$-unique query $q_{c,p} \in Q^A$, in canonical order, the decoder reads the corresponding integer $i$
inside $(h)$ and infers that $s[i+p] = c$.
Again, conflict freedom ensures that all queries can be executed and the process can terminate successfully.
Now, the string $s$ is rebuilt.

\begin{lemma}
\label{lem:queryencode}
The second step encodes information $(e)$--$(h)$ using
at most $(n-\ell) \log \sigma + O(n) - \gamma\log(\sigma/z)$ bits.
\end{lemma}
\begin{proof}
Space occupancy:~$(e)$ uses $\log{|Q^A|\choose{\gamma}} \leq |Q^A|$ bits for the subset plus,
recalling from Section~\ref{sec:technique}, $O(n)$ bits for $U$  ;~$(f)$
uses $\log{n+\gamma\choose \gamma} \leq 2n$ bits;~$(g)$ requires $(n-\ell-\gamma) \log\sigma$ bits
for the reduced unprobed region plus $\log{n-\ell\choose\gamma}$ bits for the positions removed;~$(h)$
uses $\gamma\log z + O(\gamma)$ bits.\qed
\end{proof}

\begin{proof}[of Theorem \ref{thm:generalbound}]
By combining Lemma~\ref{lem:traceencode} and Lemma~\ref{lem:queryencode} we obtain that for each $s \in S$, 
$E(s) \leq n\log\sigma + O(n) + \delta\log\left(\frac{t|Q|}{\delta\sigma}\right) - \gamma\log\left(\frac \sigma z \right).$
We know that $r + \max_{s\in S} E(s) \geq H(S) \geq n\log\sigma - \Theta(n)$, hence the bound follows.\qed
\end{proof}


%% file: s3-ub.tex
\noindent Our approach follows substantially the one in \cite{BHM07indexes}, but
uses two new ingredients, that of \emph{monotone} hashing
\cite{BelazzouguiBPV09} and \emph{succinct SB-trees}
\cite{GrossiORR09}, to achieve an improved (and in many cases optimal)
result.  We first consider these problems in a slightly different
framework and give some preliminaries.

\myparagraph{Preliminaries. } We are given a subset $T \subseteq
\left[\sigma\right]$, where $|T| = m$.  Let $R(i) = |\{ j \in T | j < i \}|$ for
any $i \in \left[\sigma\right]$, and $S(i)$ be the $i+1$st element of~$T$, for
any $i \in \left[m\right]$. 

The value of $S(R(p))$ for any $p$ is named the \emph{predecessor} of
$p$ inside $T$.  For any subset $T \subseteq \left[\sigma\right]$, given access
to $S(\cdot)$, a \emph{succinct SB-tree} \cite{GrossiORR09} is a
systematic data structure that supports predecessor queries on $T$,
using $O(|T|\log\log \sigma)$ extra bits. For any $c > 0$ such that
$|T| = O(\log^{c} \sigma)$, the succinct SB-tree supports
predecessor queries in $O(c)$ time plus $O(c)$ calls to
$S(\cdot)$. The data structure relies on a precomputed table of
$n^{1-\Omega(1)}$ bits depending only on $\sigma$,d not on $T$.

A \emph{monotone minimal perfect} hash function for
$T$ is a function $h_T$ such that $h_T(x) = R(x)$ for all $x \in T$,
but $h_T(x)$ can be arbitrary if $x \not \in T$. We need the following
result:
\begin{theorem}[\cite{BelazzouguiBPV09}]
\label{thm:mmphf}
There is a monotone minimal perfect hash function for $T$ that:\\
\indent $\bullet$~occupies $O(m \log \log \sigma)$ bits and can be evaluated in
$O(1)$ time;\\
\indent $\bullet$~occupies $O(m \log^{(3)} \sigma)$ bits and can be evaluated in
$O(\log \log \sigma)$ time.
\end{theorem}

Although function $R(\cdot)$ has been
studied extensively in the case that $T$ is given explicitly, we
consider  the situation where $T$ can only be accessed
through (expensive) calls to $S(\cdot)$. We also wish to
minimize the space used (so e.g. creating an explicit copy of
$T$ in a preprocessing stage, and then applying existing solutions,
is ruled out).  We give the following extension of known results:

\begin{lemma}
\label{lemma:rankchar}
Let $T \subseteq \left[\sigma\right]$ and $|T| = m$. Then, for any
$1 \le k \le \log \log \sigma$, there is a data structure 
that supports $R(\cdot)$ in $O(\log \log \sigma)$ time 
plus $O(1 + \log k)$ calls to $S(\cdot)$, and 
uses $O((m/k) \log \log \sigma)$ bits of space.  The
data structure uses a pre-computed table (independent
of $T$) of size $\sigma^{1-\Omega(1)}$ bits.
\end{lemma}
\begin{proof}
  We construct the data structure as follows.  We store every ($\log
  \sigma$)th element of $T$ in a y-fast trie \cite{Willard83}.  This
  divides $T$ into \emph{buckets} of $\log \sigma$ consecutive
  elements.  For any bucket $B$, we store every $k$th element of $T$
  in a succinct SB-tree.
  The space usage of the
  y-fast trie is $O(m)$ bits, and that of the succinct SB-tree is
  $O((m/k) \log \log \sigma)$ bits.

  To support $R(\cdot)$, we first perform a query on the y-fast trie,
  which takes $O(\log \log \sigma)$ time.  We then perform a query in
  the appropriate bucket, which takes $O(1)$ time by looking up a
  pre-computed table (which is independent of $T$) of size
  $\sigma^{1-\Omega(1)}$.  The query in the bucket also requires
  $O(1)$ calls to $S(\cdot)$.  We have so far computed the answer
  within $k$ keys in $T$: to complete the query for $R(\cdot)$ we
  perform binary search on these $k$ keys using $O(\log k)$ calls to
  $S(\cdot)$.
\end{proof}
\paragraph*{Supporting $\rank$ and $\select$.}
In what follows, we use Lemma~\ref{lemma:rankchar} choosing $k = 1$ and $k = \log \log \sigma$.
We now show the following result, contributing to
eq.~\eqref{eq:redundancy}. Note that the first option in
Theorem~\ref{thm:ub} has optimal index size for $t$ probes, for $t \le
\log \sigma / \log \log \sigma$.  The second option has optimal index
size for $t$ probes, for $t \le \log \sigma / \log^{(3)} \sigma$, but
only for $\select$.
\begin{theorem}
\label{thm:ub}
For any $1 \le t \le \sigma$, there exist
data structures with the following complexities:
\begin{itemize}
\vspace*{1ex}
\item[(a)] $\select$ in $O(t)$ probes and $O(t)$ time,
and $\rank$ in $O(t)$ probes and $O(t + \log \log \sigma)$
time  using a succinct index with $r=O(n (\log \log \sigma + (\log \sigma)/t))$
bits of redundancy.  If $\sigma = (\log n)^{O(1)}$, the $\rank$ operation requires
only $O(t)$ time.


\item[(b)] $\select$ in $O(t)$ probes and $O(t \log \log \sigma)$ time,
and $\rank$ in $O(t \log^{(3)} \sigma)$ probes and 
$O(t \log \log \sigma \log^{(3)} \sigma)$ time,  using 
$r=O(n (\log^{(3)} \sigma + (\log \sigma)/t))$ bits of redundancy for the succinct index. 
\end{itemize}
\end{theorem}

\begin{proof}
We divide the given string $s$ into contiguous blocks of size $\sigma$ (assume 
for simplicity that $\sigma$ divides $n = |s|$). 
As in \cite{BHM07indexes,GMR06large},  we use
$O(n)$ bits of space, and incur an additive $O(1)$-time
slowdown, to reduce the problem of supporting 
$\rank$ and $\select$ on $s$ to the problem of supporting 
these operations on a given block $B$.  We denote the
individual characters of $B$ by 
$B\left[0\right],\ldots,B\left[\sigma -1\right]$.

Our next step is also as in \cite{BHM07indexes}: letting $n_c$ denote
the multiplicity of character $c$ in $B$, we store the bitstring $Z =
1^{n_0}01^{n_1}0\ldots 1^{n_{\sigma-1}}0$, which is of length
$2\sigma$, and augment it with the binary $\rank$ and $\select$
operations, using $O(\sigma)$ bits in all.  Let $c = B\left[i\right]$ for some $0
\le i \le \sigma - 1$, and let $\pi\left[i\right]$ be the position of $c$ in a
stably sorted ordering of the characters of $B$ ($\pi$ is a
permutation).  As in~\cite{BHM07indexes}, $\select(c,\cdot)$ is
reduced, via $Z$, to determining $\pi^{-1}(j)$ for some
$j$. 
As shown in \cite{MRRR03}, for any $1 \le t \le
\sigma$, permutation $\pi$ can be augmented with $O(\sigma + (\sigma
\log \sigma) / t)$ bits so that $\pi^{-1}(j)$ can be computed in
$O(t)$ time plus $t$ evaluations of $\pi(\cdot)$ for various
arguments.

If $T_c$ denotes the set of indexes in $B$ containing the
character $c$, we store a minimal monotone hash function
$h_{T_c}$ on $T_c$, for all $c \in \left[\sigma\right]$. To compute 
$\pi(i)$, we probe $s$ to find $c = B\left[i\right]$, and observe 
that $\pi(i) = R(i) + \sum_{i=0}^{c-1} n_i$.
The latter term is obtained in $O(1)$ time by $\rank$ 
and $\select$ operations on $Z$, and the former
term by evaluating $h_{T_c}(i)$.  
By Theorem~\ref{thm:mmphf}, the complexity
of $\select(c,i)$ is as claimed.

As noted above, supporting $\rank(c,i)$ on $s$ reduces
to supporting $\rank$ on an individual block~$B$.  
If $T_c$ is as above, we apply Lemma~\ref{lemma:rankchar} 
to each $T_c$, once with $k = 1$ and 
once with $k = \log \log \sigma$.  Lemma~\ref{lemma:rankchar}
requires some calls to $S(\cdot)$, but this is just
$\select(c,\cdot)$ restricted to $B$, and is solved as 
described above.
If $\sigma = (\log n)^{O(1)}$, then $|T_c| = (\log n)^{O(1)}$,
and we store $T_c$ itself in the succinct SB-tree,
which allows us to compute $R(\cdot)$ in $O(1)$ time using a
(global, shared) lookup table of size $n^{1-\Omega(1)}$ bits.
\qed
\end{proof}
The enhancements described here also lead to more efficient
non-systematic data structures. Namely, for $\sigma = \Theta(n^\varepsilon)$ , $0 < \varepsilon < 1$, 
we match the lower bound of ~\cite[Theorem 4.3]{golynski09}. Moreover, we improve asymptotically
both in terms of space and time over the results of~\cite{BHM07indexes}:
\begin{corollary}
\label{cor:non-systematic}
There exists a data structure that represents any string $s$ of length $n$ using $nH_k(s) + O(\frac{n\log \sigma}{\log \log \sigma})$ bits, for any $k = o(\frac{\log_\sigma n}{\log \log \sigma})$, supporting $\access$ in $O(1)$ time, 
$\rank$ and $\select$ in $O(\log \log \sigma)$ time.
\end{corollary}
\begin{proof}
  We take the data structure of Theorem~\ref{thm:ub}(a), where
  $r=O(\frac{n \log \sigma}{\log \log \sigma})$. We compress $s$ using
  the high-order entropy encoder
  of~\cite{FerraginaV07,GonzalezN06,SadakaneG06} resulting in an
  occupancy of $nH_k(s) + a$ bits, where $H_k(s)$ is the $k$th-order
  empirical entropy and $a$ is the extra space introduced by encoding.
  We have $a = O(\frac {n}{\log_\sigma n}(k\log\sigma + \log \log
  n))$, which is $O(\frac{n \log \sigma}{\log \log \sigma})$ for our
  choice of $k$, hence it doesn't dominate on the data structure
  redundancy.  Operation $\access$ is immediately provided in $O(1)$
  time by the encoded structure, thus the time complexity of
  Theorem~\ref{thm:ub} applies.\qed
\end{proof}

%% file: apdx-a.tex
\normalsize
\newpage
\appendix
\section{Appendix}
\subsection{Extending previous work}
\label{apdx:extending}
In this section, we prove a first lower bound for $\rank$ and
$\select$ operations. We extend the existing techniques
of~\cite{golynski07}, originally targeted at $\sigma = 2$.  The bound
has the advantage to hold for any $1 \leq t \leq n/2$, but it is
weaker than eq.~(\ref{eq:redundancy}) when $\log t = o(\log \sigma)$.
\begin{theorem}
  \label{thm:golynskiext}
  Let $s$ be an arbitrary string of length $n$ over the alphabet
  $\Sigma=\left[\sigma\right]$, where $\sigma \leq n$. Any algorithm solving
  $\rank$ or $\select$ queries on $s$ using at most $t$ character probes
  (i.e.\mbox{} $\access$ queries), where $1 \leq t \leq n/2$, requires
  a succinct index with $r = \Omega\bigl(\frac{n \log t}{t}\bigr)$
  bits of redundancy.
\end{theorem}
Intuitively speaking, the technique is as follows: it first
creates a set of queries the data structure must answer and then partitions
the string into classes, driven by the algorithm behaviour. A bound on the
entropy of each class gives the bound. 
However, our technique proves that finding a set of queries
adaptively for each string can give an higher bound for $t = o(\log\sigma)$.

Before getting into the full details we prove a technical lemma that is based on the concept of \emph{distribution} of characters in a string:
Given a string $T$ of length $u$ over alphabet $\phi$, the distribution (vector) $\bm{d} $ for $u$ over $\phi$ is a vector in $\mathbb{N}^\phi$ containing the frequency of each
character in $T$. We can state:
\begin{lemma}
\label{lem:multinomial}
For any $\phi \geq 2$, $u \geq \phi$ and distribution $\bm{d}$ for $u$ on $\phi$,  it holds
$$\max_{\bm d} {u \choose{\bm{d}_1 \bm{d}_2 \cdots \bm{d}_{\phi}} } = \frac{u!}{\left(\frac{u}{\phi}!\right)^\phi} \leq \phi^{u}\left(\frac{\phi}{u}\right)^{\phi/2} \sqrt{\phi}.$$
\end{lemma}
\begin{proof}
The maximization follows from the concavity of the multinomial function and the uniquness of its maximum: the maximum is located at the uniform distribution $\bm d = (u/\phi, u/\phi, \ldots, u/\phi)$.
The upper bound arises from double Stirling inequality, as we have:
\begin{eqnarray*}
\frac{u!}{\left(\frac{u}{\phi}!\right)^\phi} &\leq& \frac{\sqrt{2\pi} u^{u+1/2} e^{-u + \frac{1}{12 u}}}{\left(\sqrt{2\pi}\right)^{\phi} (u/\phi)^{u+\phi/2} e^{-u + \frac{\phi}{12(u/\phi)+12}}} \\
&\leq& (2\pi)^{(1-\phi)/2}) u^{u+1/2} \left(\frac{\phi}{u}\right)^{u+1/2+(\phi-1)/2} \\
&=& O(2^{(1-\phi)/2}) \phi^{u} \phi^{1/2} \left(\frac{\phi}{u}\right)^{(\phi-1)/2}
\end{eqnarray*}
and the lemma follows.
\end{proof}
Let $L = 3\sigma t$ and assume for sake of simplicity that $L$ divides $n$. We start by defining the query set
$$\mathcal Q = \{ \select(c,{3 t i}) | c \in \left[\sigma\right] \land i \in \left[n/L\right] \}$$
having size $\gamma = \frac {n \sigma}{L} = \frac n {3t}$.
The set of strings on which we operate, $S$, is designed so that all queries in $\mathcal Q$ return a position in the set.
We build strings by concatenating $n/L$ \emph{chunks}, each of which is generated in all possible ways. A single chunk is built by aligning $3t$ occourrences of each
symbol in $\left[\sigma\right]$ and then permuting the resulting substring of length $L$ in any possible way.

A \emph{choices} tree for $\mathcal Q$ is a composition of smaller decision trees. At the top, we build the full binary tree of height $r$, 
each leaf representing a possible choices for the index bits values. For each leaf of $r$, we append the decision tree of our algorithm for the first query ${\mathcal Q}_1$
on every possible string conditioned on the choice of the index. 
The decision tree has height at most $t$ and each node has fan-out $\sigma$, being all possible results of probing a location of the string. 
Each node is labeled with the location the algorithm chooses to analyze, however we are not interested in this information. 
The decision tree has now $2^r \sigma^t$ leaves. At each leaf we append the decision tree for the second query ${\mathcal Q_2}$, increasing the number
of leaves again, and so on up to ${\mathcal Q_{\gamma}}$. 
Without loss of generality we will assume that all decision trees have height exactly $t$ and that each location is probed only once (otherwise
we simply remove double probes and add some padding ones in the end).
Leaves at the end of the whole decision tree are assigned strings from $S$ which are \emph{compatible} with the root-to-leaf path: each path
defines a set of answers $A$ for all $\gamma$ queries and a string is said to be compatible with a leave if the answers to $\mathcal Q$ on that string
is exactly $A$ and all probes during the path match the path. 
For any leaf $x$, we will denote the amount of compatible strings by $C(x)$. Note that the tree partitions the entire set
of strings, i.e. $\sum_{x \mbox{\scriptsize \ is a leaf}} C(x) = |S|$. Our objective is to prove that $C(x)$ cannot be too big, and
so prove that to distinguish all the answer sets the topmost tree must have at least some minimum height. 
More in detail, we will first compute $C^*$, an upper bound on $C(x)$ for any $x$, and then use the following relation to obtain the bound:
\begin{equation}
\label{eq:boundbase}
\log |S| = \log \sum_{x \mbox{\scriptsize \  is a leaf }} C(x) \leq \log (\mbox{\footnotesize \# of leaves}) + \log C^* \leq r+{t\gamma}\log \sigma+\log C^*
\end{equation}

Before continuing, we define some notation. For any path, the number of \emph{probed} locations is $t\gamma = n/3$, while the number of \emph{unprobed} locations is denoted by $U$.
We divide a generic string in some leaf $x$ into consecutive \emph{blocks} of characters defined depending on the answer set to $Q$ for that leaf, as follows.
The set $S_i \in \mathcal Q$ of $\sigma$ of
queries is defined as $S_i = \mathcal Q \cap \{ (c,x) \in \left[n\right] \times \left[\sigma\right] | x = i \}$; we define the block $B_i$ as the interval
$\left[\min_c A_x(S_i), (\min_c A_x(S_{i+1}))-1\right]$ (where $A_x$ defines the answer to a set of queries), i.e. the maximum span covered by answer set to $S_i$ in some leaf $x$.
Note that the partitioning in blocks is dependant only on $\mathcal Q$, i.e. it is typical of a leaf and not of a specific string, and that due
to our particular choice of $S$, the length $|B_i|$ is exactly $L$. Thus, the number of blocks is $n/L$.

We now associate a conceptual value $u_i$ to each block, which represents the
number of \emph{unprobed} characters in that block, so that $\sum_{i = 1}^{n/L} u_i = U$.
As in a leaf of the choices tree all probed locations have the same values, the only
degree of freedom distinguishing compatible strings between themselves lies in the
unprobed locations. 
We will compute $C^*$ by analyzing single blocks, and we will focus on the right side of the following:
\begin{equation}
\label{eq:cstar}
\frac{C^*}{\sigma^U} = c^*_1 c^*_2 \cdots c^*_{n/L} = \frac{g_1}{\sigma^{u_1}} \frac{g_2}{\sigma^{u_2}} \frac{g_3}{\sigma^{u_3}} \cdots \frac{g_{n/L}}{\sigma^{u_{n/L}}}
\end{equation}
where $g_i \leq \sigma^{u_i}$ represents the possible assignment of unprobed characters for block $i$ and $c^*_i$ the ratio $g_i/\sigma^{u_i}$.

We categorize blocks into two classes: \emph{determined} blocks, having $u_i < \sigma t$ and the remaining \emph{undetermined} ones.
For determined ones, we will assume $g_i = \sigma^{u_i}$. For the remaining ones we upper bound the possible choices by their maximum
value, i.e. we employ Lemma~\ref{lem:multinomial} to bound their entropy. Joining it with $u_i > \sigma t$ we obtain:
$$c^*_i \leq \frac{ \sigma^{1/2} \sigma^{u_i} \left(\frac{\sigma}{u_i}\right)^{\sigma/2} } {\sigma^{u_i}} \leq \sigma^{1/2} \left(\frac 1 t\right)^{\sigma/2}$$
The last step involves finding the number of such determined and undetermined blocks. As the number of global probes is at most $t\gamma = n/3$,
the maximum number of determined blocks (where the number of probed locations is $L-u_i > 2\sigma t$) is $(t\gamma) / (2\sigma t) = n / (2L)$.
The number of undetermined blocks is then at least $n/L - n/(2L) = n/(2L)$. Recalling that our upper bound increases with the number of
determined blocks, we keep it to the minimum. Therefore, we have:
\begin{equation}
\label{eq:cstar2}
\log C^* \leq U \log \sigma + \frac{n}{2L} \frac \sigma 2 \log \left(\frac 1 t\right) + \frac {n}{2L} \frac 1 2 \log \sigma = \Theta\left(\frac{n}{t} \log \left(\frac 1 t \right)\right)
\end{equation}
Joining Equation~\ref{eq:cstar2}, ~\ref{eq:boundbase} and the fact that $t\gamma + U = n$, we obtain that
$$n \log \sigma - \frac{n}{L} = \log |S| \leq r + t\gamma\log\sigma + U\log \sigma - \Theta\left(\frac{n}{t} \log  t\right)$$
and the bound follows.

\medskip
We can prove an identical result for operation $\rank$. The set $S$ of hard strings is the set of all strings of length $n$ over $\sigma$. We conceptually divide the strings in blocks of $L = 3\sigma t$ consecutive positions, starting at $0$. With this in mind, we define the set of queries
$$\mathcal Q = \{ \rank(c,iL) | c \in \left[\sigma\right] \land i \in \left[n/L\right] \},$$
i.e. we ask for the distribution of the whole alphabet every $L$ characters, resulting in a batch of $\gamma = \frac n {3t}$ queries. The calculations are then parallel to the previous case.
\qed